\documentclass[12pt]{amsart}

\usepackage[T1]{fontenc}
\usepackage{lmodern}
\usepackage{amsmath,amssymb,amsthm}
\usepackage{geometry,microtype}
\usepackage[colorlinks=true,linkcolor=blue,citecolor=blue,urlcolor=blue]{hyperref}
\newtheorem{theorem}{Theorem}
\newtheorem{lemma}{Lemma}
\newtheorem{corollary}{Corollary}
\newtheorem{proposition}{Proposition}

\theoremstyle{remark}
\newtheorem{observation}{Observation}
\newtheorem{remark}{Remark}

\newcommand{\E}{\mathbb E}
\newcommand{\F}{\mathbb F}
\newcommand{\C}{\mathbb C}
\newcommand{\Tr}{\operatorname{Tr}}
\newcommand{\qn}[2]{\lVert #1\rVert_{Q^{#2}}}

\title{Norm properties of quantum uniformity measures}

\author{Berke Ünal}
\address{Department of Mathematics \\ Koç University \\ Istanbul, Turkey}
\email{bunal21@ku.edu.tr}



\begin{document}

\begin{abstract}
We give a complete classification of the norm properties of quantum uniformity measures, introduced by Bu, Gu, and Jaffe to analytically characterize the Clifford hierarchy. Resolving an open question they posed, we show that quantum uniformity measures define norms precisely for $k=2,3,4$. In particular, we establish the triangle inequality for $k=4$, whereas for every $k\geq5$ it fails in some dimension. The higher-order result is proved using the structure theory of norming systems of linear equations over finite fields. We also study the geometry of quantum uniformity measures, establishing explicit bounds on the modulus of convexity whenever they define norms. Finally, we give a complete classification of the norm properties of matrix-valued Gowers uniformity functionals, showing that only $k=2$ defines a norm.
\end{abstract}

\maketitle

\section{Introduction}

The Clifford hierarchy, introduced by Gottesman
and Chuang~\cite{GC}, plays an important role in fault-tolerant quantum computation. Bu, Gu, and Jaffe~\cite{BGJ} developed a framework of quantum higher-order Fourier analysis and
introduced quantum uniformity measures $\qn{\,\cdot\,}{k}$ that characterize the Clifford
hierarchy analytically. In particular, they showed that a unitary $U$ belongs to the $k$-th level of the Clifford hierarchy if and only if $\qn{U}{k+1}=1$.

They showed that $\qn{\cdot}{2}$ and $\qn{\cdot}{3}$ are norms, and asked whether the triangle inequality holds for $k\geq4$~\cite[Remark~43]{BGJ}. We prove that $\qn{\cdot}{4}$ is a norm. We also prove that for every $k\geq5$, the triangle inequality fails.

\begin{theorem}\label{thm:main}
For every prime local dimension $p$ and every integer $n\geq1$, all linear operators $S,T$ on $\mathcal H^{\otimes n}$ satisfy
\[
  \qn{S+T}{4}\leq \qn{S}{4}+\qn{T}{4}.
\]
Consequently, $\qn{\cdot}{4}$ is a norm on the vector space of linear operators on $\mathcal H^{\otimes n}$.
\end{theorem}

\begin{theorem}\label{thm:higher-failure}
For every prime local dimension $p$ and every integer $k\geq5$, there exist an integer $n\geq1$ and linear operators $S,T$ on
$\mathcal H^{\otimes n}$ such that
\[
  \qn{S+T}{k}>\qn{S}{k}+\qn{T}{k}.
\]
\end{theorem}

This gives a complete classification of the norm properties of quantum uniformity measures. For $k=1$, only definiteness fails.
For $k=2,3,4$, the quantum uniformity measures are norms by~\cite[Proposition~39]{BGJ}, \cite[Theorem~42]{BGJ}, and Theorem~\ref{thm:main} respectively. For every $k\geq5$, only the triangle inequality fails by Theorem~\ref{thm:higher-failure}.

We also investigate the norm properties of matrix-valued uniformity functionals, studied at second order by Gowers and Hatami~\cite{GH} and at higher orders by Westdorp~\cite{Westdorp} and Jamneshan and Thom~\cite{JT}. We show that $k = 2$ is the only one that defines a norm. In particular, the triangle-inequality assertion of Jamneshan and Thom~\cite[Corollary~A.7]{JT} does not hold in the stated generality.

The paper is organized as follows. In Section~2, we recall
some properties of quantum uniformity measures and
include a remark on their limiting behavior as $k\to\infty$.
In Section~3, we introduce the quantum uniformity pairing
$\langle \cdot,\cdot\rangle_{Q^4}$, establish a Cauchy--Schwarz inequality for it, and use this to prove Theorem~\ref{thm:main}. In Section~4, we establish the generalized quantum Schwarz inequality and Hanner's inequality for the quantum uniformity measures in all cases where they define norms.

In Section~5, we recall some definitions concerning norming linear systems, following Cho et al.~\cite{CCLSV}. We then construct a family of linear systems and show that they are not norming using the structure theory of norming linear systems. In Section~6, we use this family to prove Theorem~\ref{thm:higher-failure}. Finally, in Section~7, we give a complete classification of the norm properties of
matrix-valued uniformity functionals.

\section{Preliminaries}

Throughout, let $p$ be prime and let $n\geq1$ be an integer. Let $\mathcal H=\C^p$. We work on the Hilbert space $\mathcal H^{\otimes n}$ with orthonormal basis \((e_x)_{x\in\F_p^n}\) where $e_x=e_{x_1}\otimes\cdots\otimes e_{x_n}$ for
$x=(x_1,\ldots,x_n)\in\F_p^n$.

For a linear operator $T$ on $\mathcal H^{\otimes n}$, define
its normalized trace by
\[
\tau(T)=\frac{1}{p^n}\Tr(T),
\]
and equip the vector space of linear operators on
$\mathcal H^{\otimes n}$ with the normalized Hilbert--Schmidt
inner product
\[
\langle S,T\rangle=\tau(S^*T).
\]

Throughout, we write $x\equiv y$ if $x=\lambda y$ for some $\lambda\in\C$ with $|\lambda|=1$.

\subsection{Weyl operators and quantum uniformity measures} Set $A:=\F_p^n\times\F_p^n$. For $a=(u,v)\in A$, choose Weyl operators following the phase convention of~\cite{BGJ} implicitly
\[
W_a\equiv Z^uX^v,
\]
where $X^v(e_x):=e_{x+v}$, $Z^u(e_x):=\omega^{u\cdot x}e_x$, and $\omega:=e^{2\pi i/p}$.
The Weyl operators are unitary and form an orthonormal basis of the Hilbert space of operators
on $\mathcal H^{\otimes n}$.

For $h\in A$ and a linear operator $T$ on $\mathcal H^{\otimes n}$, put $\alpha_h(T)=W_hTW_h^*$.
The Weyl relations~\cite[Section~2]{BGJ} give
\begin{equation}\label{eq:weyl}
\alpha_h(W_a)\equiv W_a.
\end{equation}
\begin{equation}\label{eq:adjoint}
W_{-a}= W_a^*.
\end{equation}

Since scalar Weyl phases disappear under conjugation, we have \(\alpha_{h+h'}=\alpha_h\alpha_{h'}\). It follows that $(\alpha_h)_{h\in A}$ is a unitary representation of the additive group $A$ on this operator Hilbert space.

The quantum derivative is defined by
\[
\partial_hX=\alpha_h(X)X^*.
\]
The quantum uniformity measures
are defined by
\begin{equation}\label{eq:definition}
  \qn{X}{k}^{2^k}
  =\E_{h_1,\ldots,h_k\in A}
     \tau\bigl(\partial_{h_k}\cdots\partial_{h_1}X\bigr).
\end{equation}
By~\cite[Proposition~26]{BGJ},
\begin{equation}\label{eq:positive}
  \qn{X}{k}^{2^k}
  =\E_{h_1,\ldots,h_{k-1}\in A}
    \left|\tau\bigl(\partial_{h_{k-1}}\cdots
    \partial_{h_1}X\bigr)\right|^2.
\end{equation}
The recursive identity~\cite[Proposition~27]{BGJ} is
\begin{equation}\label{eq:recursive}
  \qn{X}{k}^{2^k}
  =\E_{h\in A}
    \qn{\alpha_h(X)X^*}{k-1}^{2^{k-1}}.
\end{equation}

\subsection{Norm properties} We begin our discussion of the norm properties of quantum uniformity measures. For every $k\geq1$, positivity follows from~\eqref{eq:positive} and absolute homogeneity follows from~\eqref{eq:definition},
\begin{equation}\label{eq:q-homogeneity}
  \qn{\lambda X}{k}=|\lambda|\qn{X}{k}.
\end{equation}

For $k=1$, \eqref{eq:positive} gives
\[
  \qn{X}{1}=|\tau(X)|.
\]
Thus $\qn{\cdot}{1}$ is a seminorm, but fails definiteness. For example, \(\qn{W_a}{1}=0\) for all $a\neq0$.

For $k\geq2$, taking $h_1=\cdots=h_{k-1}=0$ in~\eqref{eq:positive} gives
\begin{equation}\label{eq:zero-term}
  \qn{X}{k}^{2^k}
  \geq
  p^{-2n(k-1)}
  \left|\tau\!\left((XX^*)^{2^{k-2}}\right)\right|^2.
\end{equation}
Thus $\qn{\,\cdot\,}{k}$ is definite for every $k\geq2$.
\begin{observation}\label{obs:norm-reduction}
For every $k\geq2$, the quantum uniformity measure
$\qn{\,\cdot\,}{k}$ is a norm if and only if it satisfies the
triangle inequality.
\end{observation}

The second-order Fourier formula
\cite[Proposition~39]{BGJ} gives
\begin{equation}\label{eq:q2}
  \qn{X}{2}^{4}
  =\sum_{a\in A}\left|\tau(W_a^*X)\right|^4.
\end{equation}
In particular, $\qn{\cdot}{2}$ is an $\ell^4$ norm, while
$\qn{\cdot}{3}$ is a norm by~\cite[Theorem~42]{BGJ}.
Combining~\eqref{eq:q2} with~\eqref{eq:recursive} yields
\begin{equation}\label{eq:q3}
  \qn{X}{3}^{8}
  =\frac{1}{p^{2n}}\sum_{a,b\in A}
       \left|\tau(X^*W_aXW_b)\right|^4.
\end{equation}

\begin{remark}
Although we will prove that the functionals $\qn{\,\cdot\,}{k}$ fail to be a norm
for every $k\geq5$, they nevertheless
converge to a norm as $k\to\infty$. More precisely,
\[
  \qn{\,\cdot\,}{k}\longrightarrow \|\cdot\|_{\mathrm{op}}
\]
uniformly on bounded sets as $k\to\infty$.

\begin{proof}
By~\cite[Equation~(7)]{BGJ},
\[
  \qn{X}{k}\leq\|X\|_{\mathrm{op}}.
\]
On the other hand, by positivity of $(XX^*)^{2^{k-2}}$,
\begin{align*}
  \qn{X}{k}
  &\overset{\eqref{eq:zero-term}}{\geq}
  p^{-n(k-1)/2^{k-1}}
  \tau\!\left((XX^*)^{2^{k-2}}\right)^{1/2^{k-1}}
  \notag\\
  &\geq
  p^{-n(k-1)/2^{k-1}}
  \left(
    p^{-n}
    \left\|(XX^*)^{2^{k-2}}\right\|_{\mathrm{op}}
  \right)^{1/2^{k-1}}
  \notag\\
  &=
  p^{-nk/2^{k-1}}
  \|XX^*\|_{\mathrm{op}}^{1/2}
  \notag\\
  &=
  p^{-nk/2^{k-1}}
  \|X\|_{\mathrm{op}}
\end{align*}
Hence
\[
  p^{-nk/2^{k-1}}\|X\|_{\mathrm{op}}
  \leq
  \qn{X}{k}
  \leq
  \|X\|_{\mathrm{op}}.
\]
Since $p^{-nk/2^{k-1}}\to1$, the claim follows.
\end{proof}
\end{remark}

\section{The fourth quantum uniformity measure is a norm}

In this section we prove Theorem 1. The main ingredients of the proof are the Cauchy--Schwarz inequality for the quantum uniformity pairing defined below and the triangle inequality in a finite-dimensional Bochner $L^8$ space associated with the third quantum uniformity measure.

\subsection{The quantum uniformity pairing}
The fourth quantum uniformity measure, although it is a norm as we establish in this section, is not induced by an inner product. For instance, in the single qubit case $p=2$ and $n=1$, the identity operator $I$ and the Pauli $Z$ operator fail to satisfy the parallelogram identity.

Nevertheless, we can define a pairing that behaves like an inner product inducing $\qn{\cdot}{4}$. 
For linear operators $S,T$ on $\mathcal H^{\otimes n}$, define the \emph{quantum uniformity pairing} by
\begin{equation}\label{eq:pairing}
\langle S,T\rangle_{Q^4}
:=
\left(
\E_{h\in A}
\qn{\alpha_h(S)T^*}{3}^{8}
\right)^{1/8}.
\end{equation}

Equation~\eqref{eq:recursive} implies
\begin{equation}
\langle S,S\rangle_{Q^4}=\qn{S}{4}^{2}.
\end{equation}
In particular, the pairing is positive definite.

\begin{remark}
The quantum uniformity pairing is symmetric
\[
\langle S,T\rangle_{Q^4}
=
\langle T,S\rangle_{Q^4}.
\]
Although this is not needed in the proof of
Theorem~\ref{thm:main}, it further illustrates the
inner-product-like behavior of the pairing.

\begin{proof}
First observe that
\[
\alpha_{-h}\!\left((\alpha_h(S)T^*)^*\right)
=
\alpha_{-h}(T)S^*.
\]
By~\eqref{eq:q3}, $\qn{X}{3}=\qn{X^*}{3}$, and by
~\cite[Proposition~28]{BGJ}, the quantum uniformity measures are
invariant under Weyl conjugations. Therefore,
\[
\left\lVert \alpha_h(S)T^*\right\rVert_{Q^3}
=
\left\lVert \alpha_{-h}(T)S^*\right\rVert_{Q^3}.
\]
Taking eighth powers and averaging over $h\in A$, the claim follows.
\end{proof}
\end{remark}

We now establish the Cauchy--Schwarz inequality.

\begin{lemma}\label{lem:cs}
For all linear operators $S,T$ on $\mathcal H^{\otimes n}$,
\[
|\langle S,T\rangle_{Q^4}|
\leq
\|S\|_{Q^4}\,\|T\|_{Q^4}.
\]
\end{lemma}

\begin{proof}
First, for $h,a,b\in A$, we have
\begin{align}
\alpha_h(S^*W_{-a}S)
&=\alpha_h(S^*)\alpha_h(W_{-a})\alpha_h(S)
\notag\\
&\overset{\eqref{eq:weyl}}{\equiv}
\alpha_h(S^*)W_{-a}\alpha_h(S)
\notag\\
&\overset{\eqref{eq:adjoint}}{\equiv}
\alpha_h(S^*)W_a^*\alpha_h(S).
\label{eq:weyl-expansion}
\end{align}

Taking adjoints and using trace cyclicity, we obtain
\begin{align}
\left\langle
\alpha_h(S^*W_{-a}S),T^*W_bT
\right\rangle
&=\tau\left(
\alpha_h(S^*W_{-a}S)^*T^*W_bT
\right)
\notag\\
&\overset{\eqref{eq:weyl-expansion}}{\equiv}
\tau\left(
\alpha_h(S^*)W_a\alpha_h(S)T^*W_bT
\right)
\notag\\
&=\tau\left(
T\alpha_h(S^*)W_a\alpha_h(S)T^*W_b
\right).
\label{eq:trace-identity}
\end{align}

Next, take the tensor product of two copies of the Hilbert space of operators on $\mathcal H^{\otimes n}$, equipped with the normalized Hilbert--Schmidt inner product, followed by two copies of its conjugate Hilbert space. On this tensor-product Hilbert space, define the tensor lifts
\[
Y^\sharp
=Y\otimes Y\otimes\overline Y\otimes\overline Y.
\]
Similarly, define
\[
\alpha_h^\sharp
=\alpha_h\otimes\alpha_h
\otimes\overline{\alpha_h}
\otimes\overline{\alpha_h}.
\]
These satisfy
\begin{equation}\label{eq:lift-inner}
\langle Y^\sharp,Z^\sharp\rangle
=|\langle Y,Z\rangle|^4.
\end{equation}
Moreover,
\begin{equation}\label{eq:lift-equiv}
(\alpha_h(Y))^\sharp
=\alpha_h^\sharp(Y^\sharp).
\end{equation}

Consequently, by the preceding identities,
\begin{align}
\left|\tau\left(
T\alpha_h(S^*)W_a\alpha_h(S)T^*W_b
\right)\right|^4
&\overset{\eqref{eq:trace-identity}}{=}
\left|\left\langle
\alpha_h(S^*W_{-a}S),T^*W_bT
\right\rangle\right|^4
\notag\\
&\overset{\eqref{eq:lift-inner}}{=}
\left\langle
(\alpha_h(S^*W_{-a}S))^\sharp,
(T^*W_bT)^\sharp
\right\rangle
\notag\\
&\overset{\eqref{eq:lift-equiv}}{=}
\left\langle
\alpha_h^\sharp((S^*W_{-a}S)^\sharp),
(T^*W_bT)^\sharp
\right\rangle.
\label{eq:trace-lift}
\end{align}

Since $(\alpha_h^\sharp)_{h\in A}$ is a
unitary representation, the operator
\begin{equation}\label{eq:projection}
P:=\E_{h\in A}\alpha_h^\sharp
\end{equation}
is an orthogonal projection. Hence
\begin{equation}\label{eq:projection-property}
P^2=P=P^*.
\end{equation}

For fixed $a,b\in A$, averaging
\eqref{eq:trace-lift} over $h\in A$ gives
\begin{align}
\E_{h\in A}
\left|\tau\left(
T\alpha_h(S^*)W_a\alpha_h(S)T^*W_b
\right)\right|^4
&\overset{\eqref{eq:trace-lift}}{=}
\E_{h\in A}
\left\langle
\alpha_h^\sharp((S^*W_{-a}S)^\sharp),
(T^*W_bT)^\sharp
\right\rangle
\notag\\
&\overset{\eqref{eq:projection}}{=}
\left\langle
P(S^*W_{-a}S)^\sharp,
(T^*W_bT)^\sharp
\right\rangle.
\label{eq:averaged-trace}
\end{align}

Define
\begin{equation}\label{eq:Phi}
\Phi(S):=\sum_{a\in A}(S^*W_aS)^\sharp.
\end{equation}

Summing over $a,b\in A$, we obtain
\begin{align}
\langle S,T\rangle_{Q^4}^8
&\overset{\eqref{eq:pairing}}{=}
\E_{h\in A}\qn{\alpha_h(S)T^*}{3}^{8}
\notag\\
&\overset{\eqref{eq:q3}}{=}
\frac{1}{p^{2n}}\sum_{a,b\in A}\E_{h\in A}
\left|\tau\left(
T\alpha_h(S^*)W_a\alpha_h(S)T^*W_b
\right)\right|^4
\notag\\
&\overset{\eqref{eq:averaged-trace}}{=}
\frac{1}{p^{2n}}\sum_{a,b\in A}
\left\langle
P(S^*W_{-a}S)^\sharp,
(T^*W_bT)^\sharp
\right\rangle
\notag\\
&\overset{\eqref{eq:Phi}}{=}
\frac{1}{p^{2n}}
\langle P\Phi(S),\Phi(T)\rangle
\notag\\
&\overset{\eqref{eq:projection-property}}{=}
\frac{1}{p^{2n}}
\langle P\Phi(S),P\Phi(T)\rangle.
\label{eq:pairing-representation}
\end{align}
In particular, setting $S=T=X$ in
\eqref{eq:pairing-representation} yields
\begin{equation}\label{eq:kernel-diagonal}
\|X\|_{Q^4}^{8}=\frac{1}{p^n}\|P\Phi(X)\|.
\end{equation}
Therefore, by the Cauchy--Schwarz inequality for Hilbert--Schmidt inner product,
\begin{align*}
\langle S,T\rangle_{Q^4}^8
&\overset{\eqref{eq:pairing-representation}}{=}
\frac{1}{p^{2n}}
\langle P\Phi(S),P\Phi(T)\rangle\\
&\leq
\frac{1}{p^{2n}}
\|P\Phi(S)\|\|P\Phi(T)\|\\
&\overset{\eqref{eq:kernel-diagonal}}{=}
\|S\|_{Q^4}^{8}\|T\|_{Q^4}^{8}
\end{align*}
Taking eighth roots proves the desired inequality.
\end{proof}

\subsection{Proof of Theorem~\ref{thm:main}}

Equip $A$ with the uniform probability measure and let $\mathsf{Q}_3$ denote the Banach space of linear operators on $\mathcal H^{\otimes n}$ equipped with the norm
$\qn{\cdot}{3}$.

We work in the finite-dimensional Bochner space \[L^8(A;\mathsf{Q}_3).\] Its norm is given by
\begin{equation}\label{eq:L8}
\|F\|_{L^8}
=\left(\E_{h\in A}\qn{F(h)}{3}^{8}\right)^{1/8}.
\end{equation}
Throughout, expressions depending on $h\in A$ are regarded as operator-valued functions on $A$, and their $L^8$ norms are taken over this variable.

By~\eqref{eq:L8} and~\eqref{eq:recursive},
\begin{equation}\label{eq:L8-derivative}
\|\partial_hX\|_{L^8}=\qn{X}{4}^{2},
\end{equation}
and
\begin{equation}\label{eq:kernel}
\|\alpha_h(S)T^*\|_{L^8} = \langle S,T\rangle_{Q^4}.
\end{equation}

\begin{proof}[Proof of Theorem~\ref{thm:main}]
Expanding the quantum derivative gives
\[
\partial_h(S+T)
=\partial_hS+\partial_hT
+\alpha_h(S)T^*+\alpha_h(T)S^*.
\]
By the triangle inequality in $L^8$, we obtain
\[
\|\partial_h(S+T)\|_{L^8}
\leq \|\partial_hS\|_{L^8}
+\|\partial_hT\|_{L^8}
+\|\alpha_h(S)T^*\|_{L^8}
+\|\alpha_h(T)S^*\|_{L^8}.
\]
By~\eqref{eq:L8-derivative} and~\eqref{eq:kernel}, it follows that
\[
\qn{S+T}{4}^{2}
\leq \qn{S}{4}^{2}+\qn{T}{4}^{2}
+ \langle S,T\rangle_{Q^4} + \langle T,S\rangle_{Q^4}
\]
By Lemma~\ref{lem:cs},
\[
\langle S,T\rangle_{Q^4}
+\langle T,S\rangle_{Q^4}
\leq
2\qn{S}{4}\qn{T}{4}.
\]
Combining these inequalities, we conclude that
\[
\qn{S+T}{4}^{2}
\leq \qn{S}{4}^{2}+\qn{T}{4}^{2}
+2\qn{S}{4}\qn{T}{4}
=\bigl(\qn{S}{4}+\qn{T}{4}\bigr)^2.
\]
Taking square roots proves the triangle inequality. By Observation~\ref{obs:norm-reduction}, $\qn{\cdot}{4}$ is therefore a norm, completing the proof.
\end{proof}

\section{Geometry of quantum uniformity norms}
Before turning to the proof of Theorem~\ref{thm:higher-failure}, we investigate some properties of quantum uniformity norms, that is, quantum uniformity measures of orders $k=2,3,4$. We first establish the generalized quantum Schwarz inequality for $k=4$, extending the known results for $k=2,3$. We then use this inequality to study the geometry of quantum uniformity norms. In particular, we prove an analogue of Hanner's inequality and deduce uniform convexity with explicit bounds on the modulus of convexity.

\subsection{The generalized quantum Schwarz inequality}

Throughout, for $k\geq2$, write each
$\mathbf v\in\{0,1\}^k$ uniquely as
$\mathbf v=(\mathbf u,\varepsilon)$, where
$\mathbf u\in\{0,1\}^{k-1}$ and $\varepsilon\in\{0,1\}$.
For $h\in A$, we introduce the notation
\begin{equation}
S_{\mathbf u}^{(h)}
:=
\alpha_h\!\left(S_{(\mathbf u,0)}\right)
S_{(\mathbf u,1)}^*,
\qquad
\mathbf u\in\{0,1\}^{k-1}.
\end{equation}
Equivalently to~\cite[Definition~34]{BGJ}, for $k\geq1$ define recursively
\begin{equation}\label{eq:bgj-cube-recursion}
\left\{
S_{\mathbf v}
\right\}_{\mathbf v\in\{0,1\}^{k}}
(a_k,\ldots,a_1)
:=
\begin{cases}
\,\alpha_{a_1}(S_0)S_1^*,
& k=1,\\[4pt]
\displaystyle
\left\{
S_{\mathbf u}^{(a_1)}
\right\}_{\mathbf u\in\{0,1\}^{k-1}}
(a_k,\ldots,a_2),
& k\geq2.
\end{cases}
\end{equation}
The associated $k$-th order quantum uniformity inner functional~\cite[Definition~35]{BGJ} is
\begin{equation}\label{eq:uniformity-functional}
\left\langle
\left\{
S_{\mathbf v}
\right\}_{\mathbf v\in\{0,1\}^{k}}
\right\rangle_{Q^k}
:=
\E_{a_k,\ldots,a_1\in A}
\tau\!\left(
\left\{
S_{\mathbf v}
\right\}_{\mathbf v\in\{0,1\}^{k}}
(a_k,\ldots,a_1)
\right).
\end{equation}

Bu, Gu, and Jaffe established the generalized quantum Schwarz inequality for $k=2$~\cite[Proposition~39]{BGJ} and $k=3$~\cite[Proposition~41]{BGJ}, that is, for every family $\{S_{\mathbf v}\}_{\mathbf v\in\{0,1\}^k}$ of linear operators on $\mathcal H^{\otimes n}$,
\begin{equation}\label{eq:q3-generalized-schwarz}
\left|
\left\langle
\left\{
S_{\mathbf v}
\right\}_{\mathbf v\in\{0,1\}^k}
\right\rangle_{Q^k}
\right|
\leq
\prod_{\mathbf v\in\{0,1\}^k}
\qn{S_{\mathbf v}}{k}.
\end{equation}

We now prove that this inequality also holds for $k=4$.
\begin{proposition}
Let
$\{S_{\mathbf v}\}_{\mathbf v\in\{0,1\}^4}$
be linear operators on $\mathcal H^{\otimes n}$. Then
\[
\left|
\left\langle
\left\{
S_{\mathbf v}
\right\}_{\mathbf v\in\{0,1\}^4}
\right\rangle_{Q^4}
\right|
\leq
\prod_{\mathbf v\in\{0,1\}^4}
\qn{S_{\mathbf v}}{4}.
\]
\end{proposition}

\begin{proof}
By \eqref{eq:uniformity-functional} and \eqref{eq:bgj-cube-recursion} we have
\begin{equation}\label{eq:q4-functional-recursion}
\left\langle
\left\{
S_{\mathbf v}
\right\}_{\mathbf v\in\{0,1\}^4}
\right\rangle_{Q^4}
=
\E_{h\in A}
\left\langle
\left\{
S_{\mathbf u}^{(h)}
\right\}_{\mathbf u\in\{0,1\}^3}
\right\rangle_{Q^3}.
\end{equation}
By H\"older's inequality applied to the eight factors with
exponent $8$,
\begin{equation}\label{eq:q4-holder}
\E_{h\in A}
\prod_{\mathbf u\in\{0,1\}^3}
\qn{S_{\mathbf u}^{(h)}}{3}
\leq
\prod_{\mathbf u\in\{0,1\}^3}
\left(
\E_{h\in A}
\qn{S_{\mathbf u}^{(h)}}{3}^{8}
\right)^{1/8}.
\end{equation}
Combining~\eqref{eq:q4-functional-recursion}, the case $k=3$ of~\eqref{eq:q3-generalized-schwarz}, and~\eqref{eq:q4-holder}, we obtain
\[
\left|
\left\langle
\left\{
S_{\mathbf v}
\right\}_{\mathbf v\in\{0,1\}^4}
\right\rangle_{Q^4}
\right|
\leq
\prod_{\mathbf u\in\{0,1\}^3}
\left\langle
S_{(\mathbf u,0)},
S_{(\mathbf u,1)}
\right\rangle_{Q^4}.
\]
The claim now follows from Lemma~\ref{lem:cs}.
\end{proof}

\subsection{Uniform convexity}

Hanner's classical inequality states that, for \(q\geq2\) and
\(f_1,f_2\in L^q\),
\[
    \|f_1+f_2\|_{L^q}^q
    +
    \|f_1-f_2\|_{L^q}^q
    \leq
    \left(
        \|f_1\|_{L^q}
        +
        \|f_2\|_{L^q}
    \right)^q
    +
    \left|
        \|f_1\|_{L^q}
        -
        \|f_2\|_{L^q}
    \right|^q.
\]
In particular, this inequality implies uniform convexity of
\(\bigl(L^q,\|\cdot\|_{L^q}\bigr)\).

Hanner's inequality was originally established by Hanner~\cite{Hanner1956}. Hatami later proved the corresponding inequality for scalar-valued Gowers uniformity norms~\cite{Hatami2009}.

We prove an analogous inequality for the quantum uniformity measures \(\|\cdot\|_{Q^k}\) and deduce that the corresponding spaces of linear operators on
$\mathcal H^{\otimes n}$, equipped with the norm
$\|\cdot\|_{Q^k}$, is uniformly convex

\begin{proposition}
\label{prop:hanner-Qk}
Let \(k\in \{2, 3, 4\}\) and set \(q:=2^k\). For all linear operators $S_1,S_2$ on $\mathcal H^{\otimes n}$,
\[
    \|S_1+S_2\|_{Q^k}^q
    +
    \|S_1-S_2\|_{Q^k}^q
    \leq
    \left(
        \|S_1\|_{Q^k}
        +
        \|S_2\|_{Q^k}
    \right)^q
    +
    \left|
        \|S_1\|_{Q^k}
        -
        \|S_2\|_{Q^k}
    \right|^q.
\]
\end{proposition}

\begin{proof}
For each subset $E\subseteq\{0,1\}^k$, define
\[
    S_{\mathbf v}^E:=
    \begin{cases}
        S_2, & \mathbf v\in E,\\
        S_1, & \mathbf v\notin E.
    \end{cases}
\]

We have the following estimate:
\begin{align*}
    \|S_1+S_2\|_{Q^k}^q
    +
    \|S_1-S_2\|_{Q^k}^q
    &=
    \sum_{E\subseteq\{0,1\}^k}
    \bigl(1+(-1)^{|E|}\bigr)
    \left\langle
        \left\{S_{\mathbf v}^E\right\}_{\mathbf v\in\{0,1\}^k}
    \right\rangle_{Q^k}
    \\
    &=
    2\sum_{\substack{E\subseteq\{0,1\}^k\\
                     |E|\ \mathrm{even}}}
    \left\langle
        \left\{S_{\mathbf v}^E\right\}_{\mathbf v\in\{0,1\}^k}
    \right\rangle_{Q^k}
    \\
    &\leq
    2\sum_{\substack{E\subseteq\{0,1\}^k\\
                     |E|\ \mathrm{even}}}
    \|S_1\|_{Q^k}^{q-|E|}
    \|S_2\|_{Q^k}^{|E|}
    \\
    &=
    2\sum_{\substack{0\leq j\leq q\\ j\ \mathrm{even}}}
    \binom{q}{j}
    \|S_1\|_{Q^k}^{q-j}
    \|S_2\|_{Q^k}^j
    \\
    &=
    \left(
        \|S_1\|_{Q^k}
        +
        \|S_2\|_{Q^k}
    \right)^q
    +
    \left|
        \|S_1\|_{Q^k}
        -
        \|S_2\|_{Q^k}
    \right|^q.
\end{align*}

The first equality follows by expanding the quantum uniformity inner functional in its $2^k$ arguments. The inequality follows from the generalized quantum Schwarz inequality, for $k=2,3$ by~\cite[Propositions~39 and~41]{BGJ}, and for $k=4$ by the proposition in the preceding subsection. Grouping the terms according to $j=|E|$ gives the binomial coefficient $\binom{q}{j}$. The final equality is the even part of the binomial expansion, since $q$ is even.
\end{proof}

\begin{corollary}
\label{cor:uniform-convexity-Qk}
Let $k\in\{2,3,4\}$ and set $q:=2^k$. Then the space of linear operators on $\mathcal H^{\otimes n}$, equipped with the norm $\|\cdot\|_{Q^k}$, is uniformly convex and its modulus of convexity satisfies
\[
    \delta(\varepsilon)
    \geq
    1-
    \left(
        1-
        \left(\frac{\varepsilon}{2}\right)^q
    \right)^{1/q},
    \qquad 0<\varepsilon\leq2.
\]
In particular, the norm $\|\cdot\|_{Q^k}$ is
strictly convex.
\end{corollary}

\begin{proof}
It is a standard consequence of Hanner's inequality for $q\geq2$ that the norm is uniformly convex with the stated modulus of convexity, see~\cite[Theorem~2]{Hanner1956}.
Thus the conclusion follows immediately from Proposition~\ref{prop:hanner-Qk}.
Since every uniformly convex norm is strictly convex, the final assertion follows as well.
\end{proof}

\section{Norming linear systems}

To prepare for the proof of Theorem~\ref{thm:higher-failure}, we recall some notions from the theory of norming linear systems and construct a family of linear systems that will be used in the proof.

\subsection{Basic notions} We follow the terminology of Cho et al.~\cite{CCLSV}. For a positive integer $s$, write
$[s]=\{1,\ldots,s\}$. An $m\times\ell$
\emph{linear system} over $\F_p$ is a matrix
\[
  L=(L_{ij})_{i\in[m],\,j\in[\ell]}
  \in\F_p^{m\times\ell}
\]
with linearly independent rows. We use the same symbol for the matrix and the corresponding homogeneous system of linear equations, and write
\[
  R(L)\subseteq\F_p^\ell
\]
for its row space.

For each $n\geq1$, define
\[
  \mathrm{Sol}(L)
  :=
  \left\{
    (x_1,\ldots,x_\ell)\in(\F_p^n)^\ell
    \;\middle|\;
    \sum_{j=1}^{\ell}L_{ij}x_j=0
    \text{ for every }i\in[m]
  \right\}.
\]
For a real-valued function $f:\F_p^n\to\mathbb R$, define
\[
  t_L(f)
  :=
  \E_{(x_1,\ldots,x_\ell)\in\mathrm{Sol}(L)}
  \prod_{j=1}^{\ell}f(x_j).
\]
Set
\[
  \|f\|_L:=|t_L(f)|^{1/\ell}.
\]
The system $L$ is \emph{norming} if $\|\cdot\|_L$ is a norm on the real-valued functions on $\F_p^n$ for every $n\geq1$.

\begin{remark}
The classical Gowers uniformity norms are special cases of the preceding construction. For \(d\geq2\), let \(L_{^d}\) be a linear system whose solution set consists of the \(d\)-dimensional additive cubes
\[
  \bigl(x+\omega\cdot h\bigr)_{\omega\in\{0,1\}^d},
  \qquad
  \omega\cdot h:=\sum_{i=1}^d \omega_i h_i.
\]
Since this parametrization of \(\mathrm{Sol}(L_{^d})\) by \(x,h_1,\ldots,h_d\in\F_p^n\) is bijective, for every real-valued function \(f:\F_p^n\to\mathbb R\) we have
\[
  t_{L_{^d}}(f)
  =
  \E_{x,h_1,\ldots,h_d\in\F_p^n}
  \prod_{\omega\in\{0,1\}^d}
  f\bigl(x+\omega\cdot h\bigr)
  =
  \|f\|_{U^d}^{2^d}.
\]
As \(L_{^d}\) has \(2^d\) variables, it follows that
\[
  \|f\|_{L_{^d}}=\|f\|_{U^d}.
\]
\end{remark}

Two linear systems $L_1$ and $L_2$ are \emph{isomorphic}, written $L_1\cong L_2$, if, after a permutation of the variables, they have the same row space. Equivalently, one can be obtained from the other by invertible row operations and a permutation of columns.

A linear system $L$ is \emph{indecomposable} if there do not exist linear systems $L_1$ and $L_2$, each involving a nonempty set of variables, such that
\[
  L\cong L_1\oplus L_2.
\]
If
\[
  L\cong L_1\oplus\cdots\oplus L_s,
\]
where every $L_i$ is indecomposable, then the systems $L_1,\ldots,L_s$ are called the \emph{components} of $L$. By~\cite[Theorem~4.8]{CCLSV}, any two components of a norming system are isomorphic.

\subsection{The systems $L_r$} We now introduce a family of linear systems $L_r$. Their construction is motivated by the Weyl moments introduced in Section~6.

For each even $r\geq4$, let $L_r$ be the
$r\times2r$ linear system in the variables
\[
  x_1,\ldots,x_r,y_1,\ldots,y_r
\]
defined by
\[
  \begin{aligned}
    x_i-y_i=x_r-y_r,
      &&1\leq i\leq r-1,\\
    \sum_{i=1}^{r/2}x_i
      =\sum_{i=r/2+1}^{r}x_i.
  \end{aligned}
\]
For example, for $r = 4$ 
\[
L_4:=
\begin{pmatrix}
1 & 0 & 0 & -1 & -1 & 0 & 0 & 1\\
0 & 1 & 0 & -1 & 0 & -1 & 0 & 1\\
0 & 0 & 1 & -1 & 0 & 0 & -1 & 1\\
1 & 1 & -1 & -1 & 0 & 0 & 0 & 0
\end{pmatrix}.
\]

The first $r-1$ rows are linearly independent, and every nonzero linear combination of them has a nonzero $y$-coefficient. Since the last row has zero $y$-coefficients, the $r$ rows are linearly independent.

\begin{lemma}\label{lem:higher-indecomposable}
For every even $r\geq4$, the linear system $L_r$ is indecomposable.
\end{lemma}

\begin{proof}
Observe that every nonzero vector in $R(L_r)$ has support of size at least four.

Suppose that $L_r$ is decomposable. By definition, its variables admit a nontrivial partition into two parts such that
\[
  R(L_r)=R_1\oplus R_2,
\]
where each $R_j$ is supported on the corresponding part.

For each $1\leq i<r$, let $q_i$ be the coefficient vector of
\[
  x_i-y_i-x_r+y_r=0.
\]
Since
\[
  \operatorname{supp}(q_i)=\{x_i,y_i,x_r,y_r\},
\]
this support cannot meet both parts: otherwise, the two components of $q_i$ in $R_1\oplus R_2$ would be nonzero vectors in $R(L_r)$ with support of size less than four. Hence $\operatorname{supp}(q_i)$ lies entirely in one part.

All these supports contain $x_r$ and $y_r$, so they lie in the same part. Their union contains all $2r$ variables, contradicting the nontriviality of the partition. Therefore $L_r$ is indecomposable.
\end{proof}

\subsection{The systems $L^{(k)}$} For every integer $k\geq5$, define the direct-sum system
\[
  L^{(k)}
  :=
  L_4
  \oplus
  \bigoplus_{j=2}^{k-2}L_{2^j}.
\]

The reason for the additional $L_4$ summand will become clear in Section~6, when we compute the $\qn{\,\cdot\,}{k}$ of rank one operators.

We first establish some properties of $L^{(k)}$. Since $L_r$ has $2r$ variables, $L^{(k)}$ has
\(
  8 +\sum_{j=2}^{k-2}2^{j+1}
  =
  2^k
\)
variables. Its associated functional is therefore
\begin{equation}
  \|f\|_{L^{(k)}}
  =
  |t_{L^{(k)}}(f)|^{1/2^k}.
\end{equation}
Since $L^{(k)}$ is a direct sum, its solution set factors over the summands~\cite[Section~4]{CCLSV},
\begin{equation}\label{eq:direct-sum-count}
  t_{L^{(k)}}(f)
  =
  t_{L_4}(f)
  \prod_{j=2}^{k-2}t_{L_{2^j}}(f).
\end{equation}

By Lemma~\ref{lem:higher-indecomposable}, each summand of $L^{(k)}$ is indecomposable. Hence these summands are precisely the components of $L^{(k)}$. In particular, since $k\geq5$, $L^{(k)}$ has components $L_4$ and $L_8$, which are not isomorphic since they have eight and sixteen variables, respectively. Therefore, by~\cite[Theorem~4.8]{CCLSV}, $L^{(k)}$ is not norming for every $k\geq5$. Hence, for each $k\geq5$, there exists an integer $n\geq1$ such that $\|\cdot\|_{L^{(k)}}$ is not a norm.

\section{Failure of the triangle inequality for higher orders}

We now connect the linear systems introduced in the preceding section with the quantum uniformity measures. For a function $f:\F_p^n\to\mathbb R$, define
\[
  \psi_f
  :=
  \sum_{x\in\F_p^n} f(x)e_x
  \in\mathcal H^{\otimes n}.
\]

If $\qn{\,\cdot\,}{k}$ were a norm and agreed with $\|\cdot\|_{L^{(k)}}$ up to a strictly positive constant under a linear injection, then $\|\cdot\|_{L^{(k)}}$ would also be a norm, contradicting the previous section and proving Theorem~\ref{thm:higher-failure}.

This is precisely the case for the linear injection
\[
  f\longmapsto \psi_fe_0^*.
\]
To establish the desired property, we first need to introduce the following auxiliary quantity.

For
$\psi\in\mathcal H^{\otimes n}$, define the $r$-th Weyl moment
\begin{equation}\label{eq:higher-moment-def}
  m_r(\psi)
  :=
  \E_{a\in A}
  |\langle \psi,W_a\psi\rangle|^r.
\end{equation}

We first relate these to the systems $L_r$.

\begin{lemma}\label{lem:higher-identification}
For each even $r\geq4$, there exists a constant $c>0$ such that, for every function $f:\F_p^n\to\mathbb R$,
\[
  m_r(\psi_f)=c\,t_{L_r}(f).
\]
\end{lemma}

\begin{proof}
For $a=(u,v)\in A$, the definition of the Weyl operators gives
\begin{equation}\label{eq:higher-weyl-coefficient}
  \langle \psi_f,W_{(u,v)}\psi_f\rangle
  \equiv
  \sum_{x\in\F_p^n}
  f(x)f(x-v)\omega^{u\cdot x}.
\end{equation}
Since $r$ is even, expanding the absolute $r$-th power in
\eqref{eq:higher-moment-def} and using
\eqref{eq:higher-weyl-coefficient} gives
\[
  m_r(\psi_f)
  =
  \E_{v\in\F_p^n}
  \sum_{x_1,\ldots,x_r\in\F_p^n}
  \left(\prod_{i=1}^r f(x_i)f(x_i-v)\right)
  \E_{u\in\F_p^n}\,
  \omega^{
    u\cdot\left(
      \sum_{i=1}^{r/2}x_i-
      \sum_{i=r/2+1}^{r}x_i
    \right)
  }.
\]
By character orthogonality
$\E_{u\in\F_p^n}\omega^{u\cdot z}=\delta_{z,0}$, we have
\begin{equation}\label{eq:higher-moment-expansion}
  m_r(\psi_f)
  =
  \E_{v\in\F_p^n}
  \sum_{\substack{x_1,\ldots,x_r\in\F_p^n\\
    \sum_{i=1}^{r/2}x_i
    =
    \sum_{i=r/2+1}^{r}x_i}}
  \prod_{i=1}^r f(x_i)f(x_i-v).
\end{equation}
Setting
\[
  y_i=x_i-v,
  \qquad 1\leq i\leq r,
\]
gives a bijection between the tuples appearing in
\eqref{eq:higher-moment-expansion} and $\mathrm{Sol}(L_r)$.
Therefore,
\[
  m_r(\psi_f)
  =
  \frac{1}{p^n}
  \sum_{(x_1,\ldots,x_r,y_1,\ldots,y_r)\in\mathrm{Sol}(L_r)}
  \prod_{i=1}^r f(x_i)f(y_i)
    =
  \frac{1}{p^n}|\mathrm{Sol}(L_r)|\,t_{L_r}(f)
\]
which proves the claim.
\end{proof}

\begin{lemma}
Fix an integer $k\geq5$. There exists a constant $c>0$ such that, for every function $f:\F_p^n\to\mathbb R$,
\[
  \qn{\psi_f e_0^*}{k}
  =
  c\,\|f\|_{L^{(k)}}.
\]
\end{lemma}
\begin{proof}
For $\psi,\phi\in\mathcal H^{\otimes n}$, one has
\[
  \partial_h(\psi\phi^*)
  =
  W_h\psi\phi^*W_h^*\phi\psi^*
  =
  \langle \phi,W_h^*\phi\rangle\,
  W_h\psi\psi^*.
\]
For every $j\geq3$, by absolute homogeneity~\eqref{eq:q-homogeneity} and~\cite[Proposition~28]{BGJ},
\begin{equation}
  \qn{\partial_h(\psi\phi^*)}{j-1}^{2^{j-1}}
  =
  |\langle \phi,W_h^*\phi\rangle|^{2^{j-1}}
  \qn{\psi\psi^*}{j-1}^{2^{j-1}}.
\end{equation}
Averaging over $h\in A$ and using~\eqref{eq:recursive},
\eqref{eq:adjoint}, and~\eqref{eq:higher-moment-def}, we obtain
\begin{equation}\label{eq:higher-rankone-recursion}
\qn{\psi\phi^*}{j}^{2^j}
=
m_{2^{j-1}}(\phi)
\qn{\psi\psi^*}{j-1}^{2^{j-1}}.
\end{equation}
By repeated application of
\eqref{eq:higher-rankone-recursion}, and using
\(m_r(e_0)=p^{-n}\), we obtain
\begin{equation}\label{eq:rankone-iteration}
  \qn{\psi_f e_0^*}{k}^{2^k}
  =
  \frac{1}{p^n}
  \left(
    \prod_{j=2}^{k-2}m_{2^j}(\psi_f)
  \right)
  \qn{\psi_f\psi_f^*}{2}^{4}.
\end{equation}
The remaining factor in
\eqref{eq:rankone-iteration} is
\begin{align}
  \qn{\psi_f\psi_f^*}{2}^{4}
  &\overset{\eqref{eq:q2}}{=}
  \sum_{a\in A}
  \left|
    \tau(W_a^*\psi_f\psi_f^*)
  \right|^4
  \notag\\
  &=
  \frac{1}{p^{4n}}
  \sum_{a\in A}
  \left|
    \langle\psi_f,W_a^*\psi_f\rangle
  \right|^4
  \notag\\
  &\overset{\eqref{eq:adjoint}}{=}
  \frac{1}{p^{4n}}
  \sum_{a\in A}
  \left|
    \langle\psi_f,W_a\psi_f\rangle
  \right|^4
  \notag\\
  &\overset{\eqref{eq:higher-moment-def}}{=}
  \frac{1}{p^{2n}}m_4(\psi_f).
\end{align}
Consequently,
\begin{equation}\label{eq:moment-product}
  \qn{\psi_f e_0^*}{k}^{2^k}
  =
  \frac{1}{p^{3n}}
  m_4(\psi_f)
  \prod_{j=2}^{k-2}m_{2^j}(\psi_f).
\end{equation}
Lemma~\ref{lem:higher-identification} implies there is a constant $c > 0$ such that
\[
\frac{1}{p^{3n}}
  m_4(\psi_f)
  \prod_{j=2}^{k-2}m_{2^j}(\psi_f) =  c^{2^k} \,t_{L_4}(f)
  \prod_{j=2}^{k-2}t_{L_{2^j}}(f)
\]
It follows from~\eqref{eq:moment-product} and \eqref{eq:direct-sum-count} that 
\begin{equation}
  \qn{\psi_f e_0^*}{k}^{2^k}
  =
  c^{2^k}\,t_{L^{(k)}}(f).
\end{equation}
This identity
also shows that \(t_{L^{(k)}}(f)\geq0\). Therefore,
\[
  \qn{\psi_f e_0^*}{k}
  =
  c\,
  |t_{L^{(k)}}(f)|^{1/2^k}
  =
  c\,\|f\|_{L^{(k)}},
\]
completing the proof.
\end{proof}

\begin{proof}[Proof of Theorem~\ref{thm:higher-failure}]
Fix $k\geq5$. By the preceding section, there exists an integer $n\geq1$ such that $\|\cdot\|_{L^{(k)}}$ is not a norm. Suppose, for contradiction, that $\|\cdot\|_{Q^k}$ satisfies the triangle inequality. Then, by Observation~\ref{obs:norm-reduction}, it is a norm. Since
\[
f\longmapsto \psi_f e_0^*
\]
is a linear injection, the preceding lemma implies that $\|\cdot\|_{L^{(k)}}$ is a norm, contradicting the choice of $n$. Therefore, $\|\cdot\|_{Q^k}$ does not satisfy the triangle inequality.
\end{proof}

\section{Norm properties of matrix-valued uniformity functionals}

Matrix-valued uniformity norms of order two were studied by Gowers and Hatami~\cite{GH}. Higher-order analogues on finite abelian groups were considered by Westdorp~\cite{Westdorp}, while Jamneshan and Thom~\cite{JT} considered matrix-valued uniformity norms of arbitrary order on finite groups. In this section, we give a complete classification of their norm properties.

Let $G$ be a finite group and let $N\geq1$. For
$F:A\to M_N(\C)$, write
\[
  \tau(T)=\frac{1}{N}\Tr(T),
  \qquad T\in M_N(\C),
\]
and, for $h\in G$, define
\[
  \Delta_hF(x)=F(hx)F(x)^*.
\]
Following~\cite{JT}, the $k$-th matrix-valued uniformity functional is defined by
\begin{equation}\label{eq:matrix-uk}
  \|F\|_{U^k(G)}^{2^k}
  =
  \E_{x,h_1,\ldots,h_k\in G}
  \tau\bigl(
    (\Delta_{h_k}\cdots\Delta_{h_1}F)(x)
  \bigr).
\end{equation}
Averaging in the last variable gives
\begin{equation}\label{eq:matrix-positive}
  \|F\|_{U^k(G)}^{2^k}
  =
  \E_{h_1,\ldots,h_{k-1}\in G}
  \tau\left(
  \left|
    \E_{x\in G}
    (\Delta_{h_{k-1}}\cdots\Delta_{h_1}F)(x)
  \right|^2
  \right),
\end{equation}
where $|T|^2=TT^*$.

The discussion of norm properties is therefore analogous to that for quantum uniformity measures in Section~2. Positivity follows from~\eqref{eq:matrix-positive}, while \eqref{eq:matrix-uk} gives absolute homogeneity,
\[
  \|\lambda F\|_{U^k(G)}
  =
  |\lambda|\|F\|_{U^k(G)}.
\]
For $k=1$, we have
\[
  \|F\|_{U^1(G)}
  =
  \left\|
    \E_{x\in G}F(x)
  \right\|_{2,\tau},
\]
so $\|\cdot\|_{U^1(G)}$ is a seminorm and, in general, fails definiteness.

For $k\geq2$, taking
$h_1=\cdots=h_{k-1}=e$ in~\eqref{eq:matrix-positive} gives
\begin{equation}\label{eq:matrix-zero-term}
  \|F\|_{U^k(G)}^{2^k}
  \geq
  \frac{1}{|G|^{k-1}}
  \tau\left(
    \left(
      \E_{x\in G}
      (F(x)F(x)^*)^{2^{k-2}}
    \right)^2
  \right).
\end{equation}
The right-hand side is strictly positive whenever $F\neq0$. Thus $\|\cdot\|_{U^k(G)}$ is definite for every $k\geq2$.

\begin{observation}\label{obs:matrix-norm-reduction}
For every $k\geq2$, the matrix-valued uniformity functional $\|\cdot\|_{U^k(G)}$ is a norm if and only if it satisfies the triangle inequality.
\end{observation}

For $k=2$, the triangle inequality was proved by Gowers and Hatami~\cite[Corollary~2.4]{GH}. Hence $\|\cdot\|_{U^2(G)}$ is a norm for every finite group $G$ and every $N\geq1$.

For every $k\geq5$, the failure of the triangle inequality follows immediately from Theorem~\ref{thm:higher-failure} and the pullback identity of Jamneshan and the author~\cite{JU}.

\begin{corollary}
For every integer $k\geq5$, there exist a finite abelian group $A$,
an integer $N\geq1$, and functions \(F_1,F_2:A\to M_N(\C)\) such that
\[
  \|F_1+F_2\|_{U^k(A)}
  >
  \|F_1\|_{U^k(A)}+\|F_2\|_{U^k(A)}.
\]
Moreover, $F$ and $G$ may be chosen to be Weyl orbit maps.
\end{corollary}

\begin{proof}
Fix a prime $p$. By Theorem~\ref{thm:higher-failure}, choose $n\geq1$ and operators $S,T$ on $\mathcal H^{\otimes n}$ such that
\[
  \qn{S+T}{k}>\qn{S}{k}+\qn{T}{k}.
\]
Set $A=\F_p^n\times\F_p^n$ and $N=p^n$, and let $F_S,F_T:A\to M_N(\C)$ be the corresponding Weyl orbit maps. Since the Weyl-orbit embedding $S\mapsto F_S$ is linear, $F_S+F_T=F_{S+T}$. Hence, by the pullback identity of~\cite{JU},
\[
  \|F_S+F_T\|_{U^k(A)}
  =
  \qn{S+T}{k}
  >
  \qn{S}{k}+\qn{T}{k}
  =
  \|F_S\|_{U^k(A)}+\|F_T\|_{U^k(A)}.
\]
\end{proof}

It remains to consider $k=3$ and $k=4$. The same example shows that the triangle inequality fails in both cases. Although the argument below does not directly extend to $k\geq5$, these cases have already been covered by the preceding corollary.

\begin{proposition}
For $k=3,4$, the matrix-valued Gowers functional $\|\cdot\|_{U^k}$ fails the triangle inequality, already for $G=\mathbb Z/2\mathbb Z$ and $M_2(\mathbb C)$.
\end{proposition}

\begin{proof}
For $t\in\mathbb R$, define
\[
F_t(0)=
\begin{pmatrix}
1&t\\
1&-t
\end{pmatrix},
\qquad
F_t(1)=
\begin{pmatrix}
1&-t\\
0&t
\end{pmatrix}.
\]
By direct computation,

\[
\begin{aligned}
\|F_t\|_{U^3}^{8}
&=c_1-\frac14t^2+O(t^4),\\
\|F_t\|_{U^4}^{16}
&=c_2-\frac34t^2+O(t^4),
\end{aligned}
\]
where $c_1=45/32$ and $c_2=405/64$.

Since their quadratic coefficients are negative, for every sufficiently small $t>0$ and $k=3,4$,
\[
\|F_{\pm t}\|_{U^k}
<\|F_0\|_{U^k}.
\]
But $F_t+F_{-t}=2F_0$. Hence
\[
\|F_t+F_{-t}\|_{U^k}
=2\|F_0\|_{U^k}
>\|F_t\|_{U^k}+\|F_{-t}\|_{U^k}.
\]
This proves the claim.
\end{proof}

We have therefore obtained the complete classification of the norm properties of matrix-valued uniformity functionals. For $k=1$, only definiteness fails in general. For $k=2$, the functional is a norm by~\cite[Corollary~2.4]{GH}. For $k=3,4$, the preceding example shows that the triangle inequality fails, and for every $k\geq5$ the same conclusion follows from the corollary above. Thus $k=2$ is the only order for which the matrix-valued uniformity functional defines a norm for every finite group and every matrix size.

In particular, the triangle-inequality assertion in
\cite[Corollary~A.7]{JT} does not hold in the stated
generality. The issue arises in identity (A.2) of \cite[Lemma~A.4]{JT}. The argument uses cyclicity of trace to justify a reordering of matrix factors. However, taking adjoints reverses the order of the factors, and this reversal cannot in general be achieved by cyclic permutations. Thus the claimed identity need not hold in general.

\vspace{\baselineskip}

\section*{Declaration of AI Assistance}

All arguments in the proof for $k=4$ were developed by the author. The proof strategy for $k\geq5$ failure
was also developed by the author. Having identified the need for linear systems
satisfying the Weyl moment identity in Lemma~\ref{lem:higher-identification},
the author discussed possible approaches to constructing such systems with AI.
All mathematical proofs and the exposition were written by the author.

\vspace{\baselineskip}

\end{document}